\documentclass[twocolumn,final]{IEEEtran}
\usepackage[latin9]{inputenc}
\usepackage{float}
\usepackage{amsmath}
\usepackage{amssymb}
\usepackage{graphicx}

\makeatletter

\floatstyle{ruled}
\newfloat{algorithm}{tbp}{loa}
\providecommand{\algorithmname}{Algorithm}
\floatname{algorithm}{\protect\algorithmname}

\usepackage{amsfonts}
\usepackage{algorithm}
\usepackage{amsthm}
\providecommand{\U}[1]{\protect\rule{.1in}{.1in}}
\newtheorem{theorem}{Theorem}

\newtheorem{definition}{Definition}

\newtheorem{proposition}{Proposition}

\makeatother

\begin{document}

\title{Multi-Branch Matching Pursuit\\
with applications to MIMO radar}

\author{Marco Rossi,~\IEEEmembership{Student~Member,~IEEE,} Alexander M.
Haimovich,~\IEEEmembership{Fellow,~IEEE,} \and and Yonina C. Eldar,~\IEEEmembership{Fellow,~IEEE}%
\thanks{M. Rossi and A. H. Haimovich are with New Jersey Institute of Technology,
Newark, NJ 07102, USA. (e-mail:\ marco.rossi@njit.edu, haimovic@njit.edu).
Y. C. Eldar is with Technion - Israel Institute of Technology, Haifa
32000, Israel (e-mail: yonina@ee.technion.ac.il).%
}}
\maketitle
\begin{abstract}
We present an algorithm, dubbed Multi-Branch Matching Pursuit (MBMP),
to solve the sparse recovery problem over redundant dictionaries.
MBMP combines three different paradigms: being a greedy method, it
performs iterative signal support estimation; as a rank-aware method,
it is able to exploit signal subspace information when multiple snapshots
are available; and, as its name foretells, it leverages a multi-branch
(i.e., tree-search) strategy that allows us to trade-off hardware
complexity (e.g. measurements) for computational complexity. We derive
a sufficient condition under which MBMP can recover a sparse signal
from noiseless measurements. This condition, named MB-coherence, is
met when the dictionary is sufficiently incoherent. It incorporates
the number of branches of MBMP and it requires fewer measurements
than other conditions (e.g. the Neuman ERC or the cumulative coherence).
As such, successful recovery with MBMP is guaranteed for dictionaries
that do not satisfy previously known conditions.
\end{abstract}

\section{Introduction}

\PARstart{L}{inear} inverse problems can be found throughout
engineering and the mathematical sciences. Usually these problems
are ill-conditioned or underdetermined, so that regularization must
be introduced in order to obtain meaningful solutions. Sparsity constraints
have emerged as a fundamental type of regularizer, and in the last
decade, an enormous body of work has been generated around the theory
of \textit{compressed sensing} \cite{EldarCS}. Radar has been among
the many areas where compressive sensing has found application, and
in particular, sparse recovery has been effectively applied to multiple
input multiple output (MIMO) radar \cite{Yu10,Rossi13,Strohmer12}.
To fit sparse recovery in localization applications, one generates
a grid of possible targets' locations and an associated unknown vector
of responses, such that only locations associated with targets are
non-zero. Therefore, the localization problem aims to recover the
support of such unknown vector (non-zero elements of the vector).

Compressive sensing seeks to recover an $n\times l$ matrix $\mathbf{X}$
from a small number of linear observations $\mathbf{Y}=\mathbf{AX}$
(possibly corrupted by noise), where the $m\times n$ matrix $\mathbf{A}$,
with $m\ll n$, is commonly referred to as \emph{measurement matrix}
or \emph{dictionary}, and its columns are called \emph{atoms}. While
the linear system is highly underdetermined ($m\ll n$), the inverse
problem still has a unique solution if $\mathbf{X}$ is \textit{sparse},
i.e., it has only $K$ non-zero norm rows out of $n$ (with $K\leq m\ll n$).
In this case, the problem of recovering the signal $\mathbf{X}$ from
$\mathbf{Y}$ can be cast as a non-convex combinatorial $\ell_{0}$-norm
problem, i.e., $\min_{\mathbf{X}}\left\Vert \mathbf{Y}-\mathbf{AX}\right\Vert _{F}$
$\ $s.t.$\ \left\Vert \mathbf{X}\right\Vert _{0}\leq K$, where $\left\Vert \mathbf{X}\right\Vert _{0}$
counts the number of non-zero norm rows of $\mathbf{X}$. In the following,
we will refer to rows of $\mathbf{Y}$ as \textit{measurements}, and
to the columns of $\mathbf{Y}$ as \textit{snapshots}. The $\ell_{0}$-norm
problem is known also under other names, such as sparse approximation
or highly nonlinear approximation \cite{TroppWright}, and it can
be related to the Deterministic Maximum Likelihood (DML) estimator
\cite{stoica89,vanTrees}. Both $\ell_{0}$-norm minimization and
DML problems require a multi-dimensional search with exponential complexity
\cite{donnho}, which is infeasible in practical scenarios. A core
algorithmic question arises for a given class of dictionaries, how
does one design a fast algorithm that provably recovers a $K$-sparse
input signal?

Finding conditions that guarantee correct recovery with practical
algorithms has been an active topic of research and one of the underpinnings
of compressive sensing theory. Compressive sensing theory \cite{EldarCS}
shows that it is possible to recover any $K$-sparse signal $\mathbf{X}$
using a practical algorithms (e.g., the relaxation of the $\ell_{0}$-norm
to an $\ell_{1}$-norm, called Basis Pursuit (BP) or LASSO \cite{Candes08}),
if the measurement matrix $\mathbf{A}$ satisfies specific properties.
For instance, a correct solution is guaranteed, if the matrix is sufficiently
incoherent (as measured by the cumulative coherence \cite{Troop})
or if it satisfies the restricted isometry property (RIP). Such properties
are satisfied with high probability for a wide class of random measurement
matrices (e.g. Gaussian, Bernoulli, or partial Fourier), as long as
a sufficient number of measurements is available (e.g. $m>\beta K\log n$
for some constant $\beta$) \cite{EldarCS}, but they may not hold
when the measurement matrix is structured (e.g., in MIMO radar \cite{Rossi13}).

While BP (or LASSO) has strong recovery guarantees, its complexity
is still considerably high for real-world implementations. As a result,
many other methods have been proposed, and the area is still very
active. These methods target a complexity reduction (from BP) using
sophisticated convex optimization theory concepts \cite{gradient,nesta,tfocs,RahutStromer},
graphical methods \cite{AMP}, reweighting family \cite{reweighting,WipfIterative},
the M-FOCUSS algorithm \cite{focuss}, local solutions of non-convex
relaxations, such as the $\ell_{p}$-norm (with $p<1$) \cite{Chartrand},
or simple, but effective, matching pursuit strategies (also known
as greedy algorithms) that estimate the support one index at a time.
The latter family includes Orthogonal Matching Pursuit (OMP) \cite{OMP},
Order Recursive Matching Pursuit (ORMP) (which is also known as Orthogonal
Least Squares) \cite{OLS} and Rank Aware-Orthogonal Regularized Matching
Pursuit (RA-ORMP) \cite{Davies10}. In some extensions of the matching
pursuit, at each iteration, more than one index is added to the provisional
support. Notable examples are CoSaMP \cite{CoSaMP} and IHT \cite{IHT}.
See \cite{EldarCS,TroppWright} for an overview of sparse recovery
algorithms.

This paper illustrates the MBMP algorithm, first proposed in \cite{Rossi12},
which builds upon the low complexity matching pursuit by leveraging
a multi-branch (i.e., tree-search) strategy. Similar to MBMP, matching
pursuit has been used in conjunction with tree-search strategies to
improve reconstruction performance. Tree-search strategies based on
matching pursuit are proposed in \cite{CotterTree,Schniter}, and
multi-branch generalizations of OMP appear in \cite{Flexible TreeOMP,A*OMP}.
In these works no multi-branch based recovery guarantee is provided.
Recently, another multi-branch generalization of OMP, called Multipath
Matching Pursuit (MMP), was proposed in \cite{MMP} together with
a recovery guarantee based on RIP. However, such guarantee does not
improve upon the RIP guarantee of BP. Moreover, whereas tree-search
algorithms in the literature focus on the SMV setup, MBMP addresses
the general MMV setup where, being rank aware, it takes advantage
of the signal subspace information. To avoid possible confusion, we
remark that MBMP can address the recovery of any sparse signal, as
it does not impose an additional structure on the sparse signals (e.g.,
tree-structured dictionary \cite{Jost TreeBased Dictionary}).

This work expands the literature by formulating recovery guarantees
for MBMP in a noisy setup: (\textit{i}) A sufficient condition under
which MBMP recovers any sparse signal belonging to a given support;
(\textit{ii}) A sufficient condition under which MBMP can recover
any $K$-sparse signal. Condition (\textit{i}), named Multi-Branch
Exact Recovery Condition (MB-ERC), generalizes the well-known Tropp's
ERC \cite{Troop} to a multi-branch algorithm. Condition (\textit{ii}),
named MB-coherence, generalizes Neuman ERC \cite{ERC tikhonov} to
a multi-branch algorithm. MB-coherence is met when the dictionary
is sufficiently incoherent and it provides a guideline to design the
multi-branch structure of MBMP. In contrast to other recovery guarantees
for tree-structure algorithms (e.g., MMP), both MB-ERC and MB-coherence
conditions improves the state-of-the-art in the sense that they enables
to guarantee MBMP success for dictionaries that do not satisfy previously
known conditions (e.g., ERC or Neuman ERC). Due to its ability to
trade-off measurements with computational complexity, MBMP is particularly
well suited to applications in which measurements are very expensive,
such as in radar applications where the number of measurements is
commensurate with the number of antenna elements.

The rest of the paper is organized as follows: Section \ref{Sec: Problem}
introduces the sparse recovery problem; Section \ref{Sec_MBMP} details
the proposed algorithm; in Section \ref{Sec_MBMP_guaranty}, we develop
recovery guarantees for MBMP; Section \ref{Sec_Numerical} contains
numerical results to demonstrate the potential of the MBMP algorithm
in the MIMO radar sparse localization framework. Section \ref{Sec_conc}
provides the conclusions.

The following notation is used: boldface denotes matrices (uppercase)
and vectors (lowercase); for a matrix $\mathbf{A}$, $\mathbf{A}\left(i,j\right)$
denotes the element at $i$-th row and $j$-th column. The complex
conjugate operator is $\left(\cdot\right)^{\ast}$, the transpose
operator is $\left(\cdot\right)^{T}$, the complex conjugate-transpose
operator is $\left(\cdot\right)^{H}$, and the pseudo-inverse operator
is $\left(\cdot\right)^{\dag}$. For a full rank matrix $\mathbf{X}\in\mathbb{C}^{m\times n}$
with $m\geq n$, we have $\mathbf{X}^{\dag}=\left(\mathbf{X}^{H}\mathbf{X}\right)^{-1}\mathbf{X}^{H}$.
The Frobenius norm of $\mathbf{X}$ is $\left\Vert \mathbf{X}\right\Vert _{F}$,
the $\ell_{1}$-induced norm is $\left\Vert \mathbf{X}\right\Vert _{1}\triangleq\max_{j}\sum_{i}\left\vert \mathbf{X}\left(i,j\right)\right\vert $
and the $\ell_{\infty}$-induced norm is $\left\Vert \mathbf{X}\right\Vert _{\infty}\triangleq\max_{i}\sum_{j}\left\vert \mathbf{X}\left(i,j\right)\right\vert $.
Given a set $S$ of indices, $\left\vert S\right\vert $ denotes its
cardinality, $\mathbf{\mathbf{A}}_{S}$ is the sub-matrix obtained
by considering only the columns indexed in $S$, and $\Pi_{\mathbf{A}_{S}}^{\perp}\triangleq\mathbf{I}-\mathbf{\mathbf{A}}_{S}^ {}\mathbf{\mathbf{A}}_{S}^{\dag}$
is the orthogonal projection matrix onto the null space of $\mathbf{\mathbf{A}}_{S}^{H}$.
Given two sets of indices, $S$ and $S^{\prime}$, $S\setminus S^{\prime}$
contains the indices of $S$ which are not present in $S^{\prime}$.
We define the support $S$ of a matrix $\mathbf{X}$ as the set of
non-zero norm rows indices, and we define $\left\Vert \mathbf{X}\right\Vert _{0}\triangleq\left|S\right|$.
We say that $\mathbf{X}$ is $K$-sparse if $\left\Vert \mathbf{X}\right\Vert _{0}\leq K$.

\section{Sparse Recovery Problem\label{Sec: Problem}}

In a noiseless setting, sparse recovery seeks the sparsest solution
to a linear system of equations \cite{TroppWright}:
\begin{equation}
\min_{\mathbf{x}}\left\Vert \mathbf{x}\right\Vert _{0}\text{ \ s.t. \ }\mathbf{y}=\mathbf{Ax.}\label{eq: SMV}
\end{equation}

This setup is known as Single Measurement Vector (SMV), highlighting
the fact that a single vector of measurements $\mathbf{y}$ is available.
More generally, when multiple measurement vectors have the same support,
the setting is known as Multiple Measurement Vectors (MMV) or joint
sparse. In this case, the model is $\mathbf{Y}=\mathbf{AX}$, where
$\mathbf{Y}\in\mathbb{C}^{m\times l}$ is the observed signal matrix,
$\mathbf{A}\in\mathbb{C}^{m\times n}$ is the measurement matrix and
the matrix $\mathbf{X}\in\mathbb{C}^{n\times l}$ is the unknown signal.
The unknown signal $\mathbf{X}$ is sparse since it has only $K\ll n$
non-zero norm rows. The MMV sparse recovery problem is to estimate
the sparse matrix $\mathbf{X}$. It has been shown \cite{EldarCS}
that, under certain conditions on the matrix $\mathbf{A}$ and the
sparsity $K$, the sparse matrix $\mathbf{X}$ can be recovered from
linear measurements $\mathbf{Y}$ by solving the nonconvex $l_{0}$-norm
problem:
\begin{equation}
\min_{\mathbf{X}}\left\Vert \mathbf{X}\right\Vert _{0}\text{ \ s.t. \ }\mathbf{Y}=\mathbf{AX},\label{eq: L0}
\end{equation}
where $\left\Vert \mathbf{X}\right\Vert _{0}$ counts the number of
non-zero norm rows of $\mathbf{X}$. In this work, we assume that
$\operatorname*{spark}\left(\mathbf{A}\right)>2K-\operatorname*{rank}\left(\mathbf{X}\right)+1$,
where $\operatorname*{spark}\left(\mathbf{A}\right)$ is the smallest
number of linearly dependent columns of the matrix $\mathbf{A}$.
This is a necessary and sufficient condition for $\mathbf{Y}=\mathbf{AX}$
to uniquely determine any $K$-sparse matrix $\mathbf{X}$ \cite{Davies10,Wax89}.

In the presence of noise, the measurements comply with
\begin{equation}
\mathbf{Y}=\mathbf{AX}+\mathbf{E}\label{eq: Ynoisy}
\end{equation}
where $\mathbf{E}\in\mathbb{C}^{m\times l}$ is the noise term. In
this scenario, the sparse matrix $\mathbf{X}$ can be recovered by
solving a relaxation of (\ref{eq: L0}), $\min_{\mathbf{X}}\left\Vert \mathbf{X}\right\Vert _{0}$
\ s.t. \ $\left\Vert \mathbf{Y}-\mathbf{AX}\right\Vert _{F}\leq\epsilon$.
The Frobenius norm is used when the noise is supposed to be i.i.d.
Gaussian distributed, but different norms should be used otherwise.
Other formulations can also be used: a Lagrangian formulation, $\min_{\mathbf{X}}\left\Vert \mathbf{Y}-\mathbf{AX}\right\Vert _{F}+\nu\left\Vert \mathbf{X}\right\Vert _{0}$,
or a cardinality-constrained formulation
\begin{equation}
\min_{\mathbf{X}}\left\Vert \mathbf{Y}-\mathbf{AX}\right\Vert _{F}\text{ \ s.t. \ }\left\Vert \mathbf{X}\right\Vert _{0}\leq K.\label{eq: L0_K}
\end{equation}
where the parameters $\epsilon$, $\nu$ and $K$ depend on prior
information, e.g., noise level or signal sparsity.

In the following, we detail the MBMP algorithm to address (\ref{eq: L0_K})
when the sparsity level $K$ is known. In scenarios when $K$ is unknown
and only $\epsilon$, or $\nu$, are available, MBMP can be used to
solve the Lagrangian formulation ($\min_{\mathbf{X}}\left\Vert \mathbf{Y}-\mathbf{AX}\right\Vert _{F}+\nu\left\Vert \mathbf{X}\right\Vert _{0}$),
or the residual constrained formulation ($\min_{\mathbf{X}}\left\Vert \mathbf{X}\right\Vert _{0}$
\ s.t. \ $\left\Vert \mathbf{Y}-\mathbf{AX}\right\Vert _{F}\leq\epsilon$),
with minor modifications to the algorithm's termination criteria and
support selection \cite{Rossi12b}.

While MBMP addresses problem (\ref{eq: L0_K}) for any measurement
matrix $\mathbf{A}$, in this work we focus on radar (e.g., target
localization) applications. In general, target localization consists
of two stages: detection and estimation \cite{vanTrees}. While detection
is a process that inherently relies on a single target point of view,
and deals with lower SNR levels, estimation builds on detection by
seeking to improve the accuracy of localization for detected targets.
In this work, we adopt an estimation point-of-view, which assumes
that the sparsity level $K$ (e.g., number of targets) is known, and
requires a medium to high SNR level. We formalize the latter condition
by assuming that the SNR is sufficient to guarantee that the support
of the combinatorial problem (\ref{eq: L0_K}) solution coincides
with the true support. As problem (\ref{eq: L0_K}) can be related
to the DML estimator, this assumption implies that such estimator
achieves the Cramér-Rao bound \cite{stoica89}. Our goal is to guarantee
a similar performance with reduced complexity, i.e., using MBMP.

In order to detail MBMP, it is instructive to first reformulate problem
(\ref{eq: L0_K}) in terms of the support $S$ of the solution $\mathbf{X}$.
In particular, (\ref{eq: L0_K}) is equivalent to
\begin{equation}
\min_{S}\left\Vert \Pi_{\mathbf{A}_{S}}^{\perp}\mathbf{Y}\right\Vert _{F}\ \ \text{s.t.}\ \ \left\vert S\right\vert \leq K.\label{eq: Pfocus}
\end{equation}
The reformulation follows by noticing that the minimization with respect
to $\mathbf{X}$ in (\ref{eq: L0_K}) can be separated into the minimization
with respect to the support $S$ and the minimization with respect
to the actual non-zero value of $\mathbf{X}$. In particular, assuming
that the spark condition is satisfied (i.e., $\operatorname*{spark}\left(\mathbf{A}\right)>2K-\operatorname*{rank}\left(\mathbf{X}\right)+1$),
for a given support $S$, the optimal non-zero value of $\mathbf{X}$
is given by the least square solution: $\mathbf{X}_{S}^{\ast}=\mathbf{\mathbf{A}}_{S}^{\dag}\mathbf{Y}$.
This reduces problem (\ref{eq: L0_K}) to problem (\ref{eq: Pfocus}).

\section{Multi-Branch Matching Pursuit\label{Sec_MBMP}}

Here we introduce MBMP, a multi-branch algorithm, which belongs to
the matching pursuit family and aims to solve problem (\ref{eq: Pfocus}).
We first discuss previous algorithms, and then we detail MBMP.

\subsection{Matching pursuit}

We start by providing an overview of matching pursuit \cite{EldarCS}.
This strategy starts with an empty provisional support $C=\emptyset$,
and then adds a new index to $C$ at each iteration, based on a selection
strategy. For example, in OMP, the index $g$ that maximizes $\left\Vert \mathbf{a}_{g}^{H}\Pi_{\mathbf{A}_{C}}^{\perp}\mathbf{Y}\right\Vert _{2}$
is selected. This selection strategy may be refined in two ways: a
dictionary refinement and a subspace refinement.

The \textit{dictionary refinement} applies when a non-empty provisional
support $C$ is already available. In this case, instead of using
the original dictionary's atoms, the current dictionary is projected
on the orthogonal subspace of $\mathbf{A}_{C}$, i.e., $\mathbf{\breve{a}}_{g}^{C}\triangleq\Pi_{\mathbf{A}_{C}}^{\perp}\mathbf{a}_{g}$,
and each atom is renormalized according to,
\begin{equation}
\mathbf{\bar{a}}_{g}^{C}\triangleq\left\{ \begin{array}{cc}
\mathbf{\breve{a}}_{g}^{C}/\left\Vert \mathbf{\breve{a}}_{g}^{C}\right\Vert _{2} & \text{if }\left\Vert \mathbf{\breve{a}}_{g}^{C}\right\Vert _{2}>0\\
\mathbf{0} & \text{otherwise}
\end{array}\right..\label{eq: dictionary refinement}
\end{equation}
The dictionary refinement procedure (\ref{eq: dictionary refinement})
distinguishes ORMP from OMP \cite{OMPvsOLS}: ORMP evaluates the inner
product between the residual and the modified atoms $\mathbf{\bar{a}}_{g}^{C}$,
while OMP computes the inner product using $\mathbf{\breve{a}}_{g}^{C}$.

The \textit{subspace refinement} is possible in an MMV scenario (when
$\operatorname*{rank}\left(\mathbf{X}\right)>1$). In such case, rather
than evaluating the norm of the inner product $\left\Vert \mathbf{a}_{g}^{H}\Pi_{\mathbf{A}_{C}}^{\perp}\mathbf{Y}\right\Vert _{2}$
using the residual $\Pi_{\mathbf{A}_{C}}^{\perp}\mathbf{Y}$, one
may use an orthonormal basis $\mathbf{U}$ of $\Pi_{\mathbf{A}_{C}}^{\perp}\mathbf{Y}$,
and compute $\left\Vert \mathbf{a}_{g}^{H}\mathbf{U}\right\Vert _{2}$.
The matrix $\mathbf{U}$ is also known as the signal subspace.

Depending on how refinement strategies are combined (dictionary refinement
and/or subspace refinement), four different algorithms are obtained.
Three of them have been already introduced in the literature \cite{Davies10}:
if dictionary and residual refinements are not used, we have the Simultaneous
Orthogonal Matching Pursuit (SOMP), which extends OMP to the general
MMV scenario; if dictionary refinement is not used, but residual refinement
is used, we have RA-OMP (which is not fully rank-aware); finally,
if both dictionary and residual refinements are used, we get the best
algorithm, namely RA-ORMP, which is fully rank-aware. In particular,
the so-called ``rank awareness'' means that, assuming $\operatorname*{spark}\left(\mathbf{A}\right)>2K-\operatorname*{rank}\left(\mathbf{X}\right)+1$
and considering a noiseless scenario, whenever the received signal
$\mathbf{Y}$ is full rank, RA-ORMP recovers the correct support with
probability one.
\begin{figure}[ptb]
\centering{}\includegraphics[width=3.2102in,height=1.3638in]{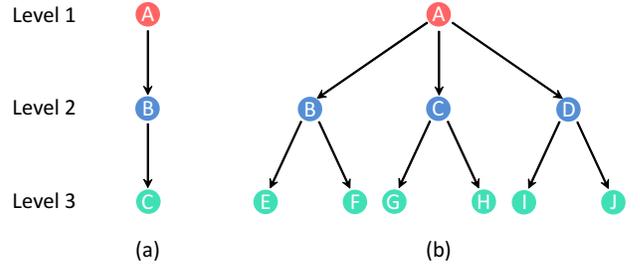}\caption{Graph of MBMP algorithm: (a) for a branch vector $\mathbf{d}=\left[1,1\right]$
- MBMP reduces to RA-ORMP; (b) for a branch vector $\mathbf{d}=\left[3,2\right]$.}
\label{fig: tree}
\end{figure}

\subsection{MBMP}

The proposed MBMP algorithm generalizes RA-ORMP by including a multi-branch
structure. In particular, it is possible to visualize RA-ORMP as a
chain of nodes, depicted in Fig. \ref{fig: tree}-(a). Node A is tagged
with an empty support. A new index is selected following a chosen
selection strategy, and it becomes the provisional support of node
B. To solve (\ref{eq: Pfocus}), this procedure is repeated until
level $K+1$ is reached.

Instead of a chain of nodes, the MBMP algorithm may be visualized
as a tree of nodes as shown in Fig. \ref{fig: tree}-(b), where each
node is allowed to have multiple children (node A is the parent of
nodes B, C and D; B is the parent of E and F). For instance, in Fig.
\ref{fig: tree}-(b), node A has $3$ branches, resulting in $3$
nodes at level $2$. Node A is tagged with an empty support. Then,
the index $g$ that maximizes $\left\Vert \mathbf{a}_{g}^{H}\mathbf{U}\right\Vert _{2}$
(where $\mathbf{U}$ is the signal subspace) becomes the provisional
support of node B. While RA-ORMP doesn't have any other node at level
2, with MBMP, the index $g$ that gives the \textit{second} largest
value of $\left\Vert \mathbf{a}_{g}^{H}\mathbf{U}\right\Vert _{2}$
is assigned to the provisional support of node C. Similarly, the index
$g$ that gives the \textit{third} largest value of $\left\Vert \mathbf{a}_{g}^{H}\mathbf{U}\right\Vert _{2}$
is assigned to the provisional support of node D. One of these atom
indices will necessarily be part of the solution returned by the algorithm.
Then, MBMP continues to populate nodes at level $3$. For example,
consider node B. Since node B has two branches, it has \textit{two}
children. Following the selection strategy, \textit{two} new indices
are selected. Each of these is added to the provisional support of
node B and used to tag node E and F, respectively. This procedure
is performed for all nodes at level $2$ (i.e., nodes C and D), thus
populating nodes G, H, I and J. The process stops when all nodes at
level $K+1$ have been populated. The support $C$ achieving the minimum
value of $\left\Vert \Pi_{\mathbf{A}_{C}}^{\perp}\mathbf{Y}\right\Vert _{F}$
is elected as the solution to (\ref{eq: Pfocus}).

The MBMP tree depends on the number of levels and on the number of
branches at each level (assumed constant for nodes within the same
level of the tree). The MBMP structure can be specified using a vector
$\mathbf{d}\triangleq\left[d_{1},\ldots,d_{K}\right]$ referred to
as \textit{branch vector}: $d_{i}$ represents the number of branches
of each node at level $i$. For instance, the tree in Fig. \ref{fig: tree}-(a)
has $\mathbf{d}=\left[1,1\right]$ while the tree in Fig. \ref{fig: tree}-(b)
has $\mathbf{d}=\left[3,2\right]$ (node A at level $1$ has $d_{1}=3$
branches, and each node at level $2$ (i.e., B, C, and D) possesses
$d_{2}=2$ branches). We call \textit{root} node the node at level
$1$ (i.e., node A in Fig. \ref{fig: tree}), and $\mathbf{U}=\operatorname{orth}\left(\mathbf{Y}\right)$
denotes an estimate of the signal subspace (see \cite{vanTrees,Davies10}
for an overview of signal subspace estimation).

The pseudo-code of the MBMP algorithm is detailed in the following
table.
\begin{algorithm}
\caption{Multi-branch matching pursuit algorithm}

\textbf{Input}: $\mathbf{Y}\in\mathbb{C}^{m\times l}$, $\mathbf{A}\in\mathbb{C}^{m\times G}$,
and $\mathbf{d}\in\mathbb{N}^{K}$

\textbf{Output}: Support of approximate solution to problem (\ref{eq: Pfocus})

1:\ \ Initialize root node (tagged with $C=\emptyset$ and $\bar{C}=\emptyset$)

2:\ \ Set $\mathbf{U}=\mathbf{Y}$, $f_{opt}=+\infty$

3: \ \textbf{for }$\forall$ node without children at level $i\leq K$

4: \ \qquad{}\textbf{if} $l>1$: Set $\mathbf{U}=\operatorname{orth}\left(\Pi_{\mathbf{A}_{C}}^{\perp}\mathbf{Y}\right)$

5: \ \qquad{}\textbf{for} $j=1,\ldots,d_{i}$

6: \qquad{}\qquad{}\textbf{$\hat{g}_{j}\in\arg\max_{g\notin\bar{C}\cup\left[\hat{g}_{1},\ldots,\hat{g}_{j-1}\right]}\left\Vert \mathbf{U}^{H}\mathbf{\bar{a}}_{g}^{C}\right\Vert _{2}$}

7: \qquad{}\qquad{}Tag a new child node with:

\qquad{}\qquad{}\qquad{}\quad{}$C=C\cup\hat{g}_{j}$, $\bar{C}=\bar{C}\cup\left[\hat{g}_{1},\ldots,\hat{g}_{j}\right]$

8: \qquad{}\qquad{}\textbf{if} $\left|C\right|=K$ and $\left\Vert \Pi_{\mathbf{A}_{C}}^{\perp}\mathbf{Y}\right\Vert _{F}<f_{opt}$:

\qquad{}\qquad{}\qquad{}\quad{}Set $S=C$, and $f_{opt}=\left\Vert \Pi_{\mathbf{A}_{C}}^{\perp}\mathbf{Y}\right\Vert _{F}$

9: \ \qquad{}\textbf{end}

10: \textbf{end}

11: Return support $S$ \label{algo1} 
\end{algorithm}

We finally note that nodes at level $K$ need only $d_{K}=1$ branch.
This is because any additional branch would be tagged with provisional
support $C$ that cannot minimize the objective function of problem
(\ref{eq: Pfocus}).

\subsection{Computational Complexity}

Given an $m\times n$ matrix $\mathbf{A}$ and an $m\times l$ matrix%
\footnote{We consider $l\leq m$. When $l>m$, we can substitute $\mathbf{Y}$
with any square root of $\mathbf{YY}^{H}$ (an $m\times m$ matrix)
without changing problem (\ref{eq: Pfocus}).%
} $\mathbf{Y}$,\textbf{ }the computational requirements of MBMP depend
on the specific implementation details, the structure of the measurement
matrix $\mathbf{A}$ and the branch vector $\mathbf{d}=\left[d_{1},\ldots,d_{K}\right]$
(MBMP has $1$ node at level $1$ and ${\textstyle \prod\nolimits _{j<i}}d_{j}$
nodes at level $i$). Due to the variability of the computation costs
of applying the transform $\mathbf{A}^{H}$ (ranging from $\mathcal{O}\left(n\log\left(m\right)l\right)$
for an FFT-type operations to $\mathcal{O}\left(nml\right)$ for unstructured
matrices), we denote with $F$ the computational cost associated with
performing $\mathbf{A}^{H}\mathbf{U}$ without specifying an associated
number of flops. Furthermore, to perform residual refinement, a practical
implementation of MBMP would also need to incorporate an estimate
of the signal subspace, and we denote $R$ the relative cost. For
a node at level $i$, other operations performed by MBMP are: selecting
the $d_{i}$ largest inner products, which is known as the ``selection
problem'' \cite{combOpt} and can be solved using $\mathcal{O}\left(n\right)$
flops; the dictionary refinement, which costs $2m\left(n-i+1\right)$
flops; the update of the projection matrix $\Pi_{\mathbf{A}_{S}}^{\perp}$,
which requires $2m$ flops, and the computation of the residual, that
needs $ml$ flops. An efficient implementation of both the dictionary
refinement and the projection matrix update is obtained by applying
a QR factorization \cite{EldarCS}.

Summarizing, the first node requires $F+R+\mathcal{O}\left(n\right)$
flops, since the dictionary refinement and the projection matrix update
are not performed. Any node at level $i$ (with $2\leq i\leq K$)
requires $F+R+\mathcal{O}\left(n\right)+2m\left(n-i+2\right)+ml$
flops. Finally, a node at level $K+1$ requires $m\left(l+2\right)$
flops to update the projection matrix and to compute the residual
norm.

As a rule of thumb, the complexity of MBMP scales approximately with
the number of nodes in the first $K$ levels of MBMP tree (i.e., all
nodes except those at level $K+1$). Therefore, while RA-ORMP complexity
is proportional to $K$, the complexity of MBMP with branch vector
$\mathbf{d}=\left[d_{1},\ldots,d_{K}\right]$ scales approximately
with $1+\sum_{i=2}^{K}{\textstyle \prod\nolimits _{j<i}}d_{j}$. For
example, the complexity of MBMP with branch vector $\mathbf{d}=\left[2,2,2,2,1\right]$
is approximately $31/5=6.2$ times that of RA-ORMP. This aspect will
be further investigated in the numerical results. It is worth mentioning
that, due to the tree-structure, the MBMP algorithm lends itself to
a parallel implementation. Indeed, if multiple processors are available,
although the total number of MBMP operations remains the same, most
of them can be performed in parallel, reducing the total algorithm's
execution time.

\section{Recovery guarantees for MBMP\label{Sec_MBMP_guaranty}}

In this section, we develop recovery guarantees for MBMP. Throughout
this section, the measurement matrix $\mathbf{A}$ is a given deterministic
matrix. MBMP is executed with a branch vector $\mathbf{d}\triangleq\left[d_{1},\ldots,d_{K}\right]$
of length $K$. The information available to the recovery algorithm
includes $\mathbf{Y}$, $\mathbf{A}$, and $K$. Moreover, as MBMP
is executed, provisional supports, denoted $C_{i}$, are available
at all nodes of level $i$. By convention, $C_{1}=\emptyset$, since
at level $1$, no provisional support is available. Finally, we say
that MBMP succeeds in recovering a $K$-sparse $\mathbf{X}$ if one
node at level $K+1$ is tagged with the correct support of $\mathbf{X}$,
denoted with $S^{*}$, which is assumed to be the (global optimal)
solution of problem (\ref{eq: Pfocus}).

The road map of this section is as follows: We start by reviewing
Tropp's ERC \cite{Troop}. This condition considers signals with a
\textit{specific} support $S^{\ast}$. This restriction enables to
obtain recovery guarantees for pursuit algorithms (e.g. BP, OMP, ORMP,
and RA-ORMP). By generalizing ERC to a multi-branch algorithm, we
formulate in Definition \ref{def: MB-ERC} the MB-ERC. Theorem \ref{th: MBMP recovery}
relies on the MB-ERC to provide a sufficient condition that guarantees
successful recovery with MBMP. Similar to ERC, MB-ERC is non-constructive,
since it focuses only on signals with a specific support $S^{\ast}$.
To overcome this limitation, in Definition \ref{def: d cum cohe cond}
we introduce the MB-coherence condition for multi-branch algorithms.
Using the MB-coherence, Theorem \ref{th: MBMP recovery2} specifies
a sufficient condition that guarantees the recovery of \textit{any}
$K$-sparse signal $\mathbf{X}$ using MBMP. Interestingly, in the
noiseless setup, the MB-coherence condition can be seen as the multi-branch
generalization of the Neuman ERC (or weak ERC) \cite{ERC tikhonov},
which improves upon the cumulative coherence condition proposed in
\cite{Troop}.

\subsection{MB-ERC}

We first overview the ERC, which characterizes the ability of practical
algorithms to recover sparse signals supported on a \textit{specific}\ support
$S^{\ast}$. For a given support $S^{\ast}$ and for a matrix $\mathbf{A}$,
the ERC is formulated \cite{Troop}
\begin{equation}
\max_{g\notin S^{\ast}}\left\Vert \mathbf{A}_{S^{\ast}}^{\dag}\mathbf{a}_{g}\right\Vert _{1}<1.\label{eq: ERC}
\end{equation}

This condition addresses linear systems of equations of the form $\mathbf{A}_{S^{\ast}}\mathbf{x}=\mathbf{a}_{g}$,
where $\mathbf{a}_{g}$ is a column from $\mathbf{A}$ that is outside
the support $S^{\ast}$. The ERC states that the minimum ($\ell_{2}$-)energy
solution to all these systems should have an $\ell_{1}$-length smaller
than $1$. The importance of the ERC stems from its strong connection
to the success of pursuit techniques. In particular, ERC is a sufficient
condition for successful recovery via RA-ORMP (as shown in \cite{Davies10})
and thus for MBMP as well. ERC is also sufficient for correct recovery
via OMP, ORMP and BP in the SMV setup (see \cite{Troop} and \cite{ERC joint ksteps}).

Next, we proceed to introduce MB-ERC, which generalizes ERC to a multi-branch
algorithm and leads to a stronger sufficient condition to guarantee
the success of MBMP. In contrast to RA-ORMP, in which each node has
only one child, the number of children of each node of MBMP is specified
by the branch vector $\mathbf{d}=\left[d_{1},\ldots,d_{K}\right]$,
where $d_{i}$ is the \textit{number of branches} of each node at
level $i$. As a result, MB-ERC is a function of $d_{i}$. To proceed,
it is convenient to define the $d_{i}$\_$\max$ operator. Explicitly,
given a positive integer $d_{i}$ and a real vector $\mathbf{z}$
(where its elements are indexed by $g$), $d_{i}$\_$\max_{g\notin S^{\ast}}\left(\mathbf{z}\right)$
is the $d_{i}$-largest entry among the indices of $\mathbf{z}$ outside
the support $S^{\ast}$. For instance, if $\mathbf{z}=\left[.7,1.4,1.1,.8,.9\right]^{T}$
and $S^{\ast}=\left\{ 2\right\} $, then $1$\_$\max_{g\notin S^{\ast}}\left(\mathbf{z}\right)=1.1$
(the largest entry outside $S^{\ast}$), while $2$\_$\max_{g\notin S^{\ast}}\left(\mathbf{z}\right)=.9$
(the second largest entry outside $S^{\ast}$), and so on.

Now we are ready to define MB-ERC. Consider level $i$ of MBMP. Given
a provisional support $C_{i}$, the dictionary refinement modification
is implemented, and we define $\mathbf{\bar{A}}^{C_{i}}\triangleq\left\{ \mathbf{\bar{a}}_{g}^{C_{i}},g\notin C_{i}\right\} $
as the resulting measurement matrix (see also (\ref{eq: dictionary refinement})).
Denoting $S\triangleq S^{\ast}\setminus C_{i}$ the support's indices
yet to be identified, we consider a sub-matrix of $\mathbf{\bar{A}}^{C_{i}}$
obtained by collecting only atoms $\mathbf{\bar{a}}_{g}^{C_{i}}$
belonging to $S$, i.e., $\mathbf{\bar{A}}_{S}^{C_{i}}\triangleq\left\{ \mathbf{\bar{a}}_{g}^{C_{i}},g\in S\right\} $.
We further define the Out-support In-support energy Ratio (OIR) as
\begin{equation}
\textrm{OIR}\triangleq\frac{\max_{g\notin S}\left\Vert \mathbf{U}^{H}\Pi_{\mathbf{A}_{S}}^{\perp}\mathbf{\bar{a}}_{g}^{C_{i}}\right\Vert _{2}}{\max_{g\in S}\left\Vert \mathbf{U}^{H}\mathbf{\bar{a}}_{g}^{C_{i}}\right\Vert _{2}}.\label{eq: OIR}
\end{equation}
In the SMV setup, we assume by convention that $\mathbf{U=y}$, while,
in the MMV setup, $\mathbf{U}=\operatorname{orth}\left(\Pi_{\mathbf{A}_{C_{i}}}^{\perp}\mathbf{Y}\right)$
is an estimate of the signal subspace given a provisional support
$C_{i}$. OIR is the square-root of the ratio between the largest
energy of $\mathbf{U}^{H}\Pi_{\mathbf{A}_{S}}^{\perp}\mathbf{\bar{a}}_{g}^{C_{i}}$
among indices \textit{outside} $S\triangleq S^{\ast}\setminus C_{i}$
and the largest energy of $\mathbf{U}^{H}\mathbf{\bar{a}}_{g}^{C_{i}}$
over the indices \textit{inside} $S$. Since the definition of the
OIR depends on unknown quantities (e.g., the support $S$), it must
be estimated.

\begin{definition}[MB-ERC]\label{def: MB-ERC}Consider a support
$S^{\ast}$, a matrix $\mathbf{A}$, a positive integer $d_{i}$,
and a correct provisional support $C_{i}\subset S^{\ast}$. Let $S\triangleq S^{\ast}\setminus C_{i}$
be the set of indices yet to be identified. The MB-ERC($S^{*},C_{i},d_{i}$)
is defined as
\begin{equation}
d_{i}\_\max_{g\notin S}\left(\left\Vert \left(\mathbf{\bar{A}}_{S}^{C_{i}}\right)^{\dag}\mathbf{\bar{a}}_{g}^{C_{i}}\right\Vert _{1}\right)<1-\textrm{OIR},\label{eq: Rossi condition}
\end{equation}
where OIR is defined in (\ref{eq: OIR}).

\end{definition}

MB-ERC generalizes ERC to a multi-branch algorithm and to a noisy
setup. In particular, in a noiseless setup ($\textrm{OIR}=0$), MB-ERC($S^{*},\emptyset,1$)
(i.e., (\ref{eq: Rossi condition}) at level $1$ with $d_{1}=1$
branches) reduces to ERC in (\ref{eq: ERC}). By using MB-ERC, we
can guarantee success of MBMP for any signal $\mathbf{X}$ supported
on $S^{\ast}$:

\begin{theorem}[Recovery of any signal supported on $S^{\ast}$]\label{th: MBMP recovery}Let
$\mathbf{X}$ be an unknown, $K$-sparse matrix of rank $r$, with
known support $S^{\ast}$, and $\mathbf{A}$ be full rank with normalized
columns and $\operatorname*{spark}\left(\mathbf{A}\right)>2K-r+1$.
Let $\mathbf{Y}=\mathbf{AX+E}$ be the noisy data with OIR given by
(\ref{eq: OIR}). If the MB-ERC in (\ref{eq: Rossi condition}) is
met for all nodes at levels $i=1,\ldots,K-1$, then MBMP with branch
vector $\mathbf{d}=\left[d_{1},\ldots,d_{K-1},1\right]$ is guaranteed
to recover $\mathbf{X}$ successfully. \end{theorem}

\begin{proof} See Appendix A. \end{proof}

Theorem \ref{th: MBMP recovery} formulates a sufficient condition
for MBMP successful recovery of sparse signals supported on a \textit{specific}
support $S^{\ast}$. In the next subsection, by removing the knowledge
of $S^{\ast}$, we obtain a condition that guarantees MBMP successful
recovery for \textit{any} $K$-sparse signal.

\subsection{MB-coherence condition}

A disadvantage of both MB-ERC and ERC is that they require the knowledge
of the\ true support $S^{\ast}$, hardly available in practice. This
implies that to check if a measurement matrix $\mathbf{A}$ satisfies
MB-ERC (or ERC), one has to compute the conditions for all $\binom{n}{K}$
possible supports $S^{\ast}$ of cardinality $K$, which is usually
prohibitive even for small values of $K$. To overcome this limitation,
we develop a practical condition that guarantees recovery via MBMP
for \textit{any} $K$-sparse signal $\mathbf{X}$. The main problem
with MB-ERC and ERC is the presence of the pseudo-inverse. As shown
in \cite{Troop}, by using standard norm inequalities to upper bound
ERC, it is possible to obtain practical conditions that include only
inner products rather than the pseudo-inverse operator. These conditions
rely on the notion of \textit{coherence} of a measurement matrix $\mathbf{A}$,
defined as $\mu\left(\mathbf{A}\right)\triangleq\max_{i\neq j}\left\vert \mathbf{a}_{i}^{H}\mathbf{a}_{j}\right\vert $
\cite{Troop}, and on the notion of cumulative coherence (also known
as \textit{Babel's function }\cite{Elad}), defined as $\bar{\mu}\left(K,\mathbf{A}\right)\triangleq\max_{S,\left\vert S\right\vert =K}\max_{g\notin S}\left\Vert \mathbf{A}_{S}^{H}\mathbf{a}_{g}\right\Vert _{1}$
\cite{Troop}. Using these definitions, it was shown in \cite{Troop}
that the ERC holds for any $K$-sparse signal $\mathbf{X}$, if either
the coherence satisfies
\begin{equation}
\mu\left(\mathbf{A}\right)<\frac{1}{2K-1}\label{eq: cohe}
\end{equation}
or if the cumulative coherence satisfies
\begin{equation}
\bar{\mu}\left(K-1,\mathbf{A}\right)+\bar{\mu}\left(K,\mathbf{A}\right)<1.\label{eq: cum cohe}
\end{equation}

A condition that requires fewer measurements is called Neuman ERC
(or weak ERC). It was proposed in \cite{ERC tikhonov}, and can be
stated as:
\begin{equation}
\max_{S,\left\vert S\right\vert =K}\left(\max_{g\in S}\left\Vert \mathbf{A}_{S}^{H}\mathbf{a}_{g}\right\Vert _{1}+\max_{g\notin S}\left\Vert \mathbf{A}_{S}^{H}\mathbf{a}_{g}\right\Vert _{1}\right)<2\left(1-\textrm{NSR}\right),\label{eq: condition 1 cohe}
\end{equation}
where the Noise-to-Signal Ratio (NSR) is defined in \cite{ERC tikhonov}.
Similarly to the OIR, NSR depends on unknown quantities (e.g., signal
and noise realizations) and must be estimated. As shown in \cite{ERC tikhonov},
condition (\ref{eq: condition 1 cohe}) may be used to guarantee correct
recovery of any $K$-sparse signal using BP.

The number of measurements required to guarantee correct recovery
can be further reduced by capturing the multi-branch structure of
MBMP. Indeed, we now develop a condition, dubbed MB-coherence, which
guarantees recovery of any $K$-sparse signal using MBMP, while requiring
less measurements than (\ref{eq: condition 1 cohe}) for a multi-branch
algorithm. Considering a provisional support $C_{i}$, as before,
we denote $\mathbf{\bar{A}}^{C_{i}}\triangleq\left\{ \mathbf{\bar{a}}_{g}^{C_{i}},g\notin C_{i}\right\} $
the associated refined measurement matrix. For the sake of notation,
in the definition, we drop the superscript $C_{i}$ from $\mathbf{\bar{A}}_{S}^{C_{i}}$
and we use $\mathbf{\bar{A}}_{S}$ instead.

\begin{definition}[MB-coherence]\label{def: d cum cohe cond}Consider
a matrix $\mathbf{A}$, integers $K$ and $d_{i}$, a provisional
support $C_{i}$, and OIR defined in (\ref{eq: OIR}), with $\textrm{OIR}<1$.
Let $k\triangleq K-\left\vert C_{i}\right\vert $. The MB-coherence($C_{i},d_{i}$)
is defined as
\begin{equation}
\max_{S,\left\vert S\right\vert =k}\left(\max_{g\in S}\left\Vert \mathbf{\bar{A}}_{S}^{H}\mathbf{\bar{a}}_{g}^{C_{i}}\right\Vert _{1}+\frac{d_{i}\text{\_}\max\limits _{g\in S}\left(\left\Vert \mathbf{\bar{A}}_{S}^{H}\mathbf{\bar{a}}_{g}^{C_{i}}\right\Vert _{1}\right)}{1-\textrm{OIR}}\right)<2.\label{eq: condition d cohe}
\end{equation}

\end{definition}

A key aspect of the MB-coherence condition is that it includes only
inner products among columns of the matrix $\mathbf{\bar{A}}^{C_{i}}$
(as opposed to MB-ERC in (\ref{eq: Rossi condition}) which incorporates
the pseudo-inverse operator). This enables to practically compute
the smallest integer $d_{i}$ such that the MB-coherence condition
(\ref{eq: condition d cohe}) is met, as discussed in Appendix C.

By using the MB-coherence condition, it is possible to obtain a sufficient
condition to guarantee that MBMP recovers any $K$-sparse signal $\mathbf{X}$:

\begin{theorem}[Recovery of \textit{any} $K$-sparse signal]\label{th: MBMP recovery2}Let
$\mathbf{X}$ be an unknown, $K$-sparse matrix of rank $r$, and
$\mathbf{A}$ be full rank with normalized columns and $\operatorname*{spark}\left(\mathbf{A}\right)>2K-r+1$.
Let $\mathbf{Y}=\mathbf{AX+E}$ be the noisy data with OIR given by
(\ref{eq: OIR}). If the MB-coherence condition in (\ref{eq: condition d cohe})
is met for all nodes at levels $i=1,\ldots,K-1$, then MBMP with branch
vector $\mathbf{d}=\left[d_{1},\ldots,d_{K-1},1\right]$ is guaranteed
to recover $\mathbf{X}$ successfully. \end{theorem}

\begin{proof} See Appendix B. \end{proof}

Theorem \ref{th: MBMP recovery2} guarantees correct recovery of any
$K$-sparse signal using MBMP. Furthermore, in a noiseless case (when
OIR$=$NSR$=0$), MB-coherence($\emptyset,1$) (i.e., (\ref{eq: condition d cohe})
at level $1$ with $d_{1}=1$ branches) reduces to (\ref{eq: condition 1 cohe}).
Since the $d\text{\_}\max$ operator is decreasing in $d$, MB-coherence($\emptyset,d_{1}$)
with $d_{1}>1$ guarantee MBMP success for dictionaries that do not
satisfy Neuman ERC. In the numerical results section, this point with
be further explored.

\subsection{Discussion}

A key aspect highlighted by the theoretical results above is that
increasing the number of branches of MBMP does not only allow us to
reduce the number of measurements, but it enables to tolerate higher
noise levels. In particular, consider MB-ERC in (\ref{eq: Rossi condition})
(MB-coherence in (\ref{eq: condition d cohe})). Since the $d_{i}\text{\_}\max$
operator is decreasing in $d_{i}$, one can preserve the validity
of MB-ERC (MB-coherence) even if the noise level increase (i.e., larger
OIR) by increasing $d_{i}$. This point will be further analyzed in
the numerical results section.

Additionally, Theorem \ref{th: MBMP recovery} (Theorem \ref{th: MBMP recovery2})
reads as the intersection of the conditions MB-ERC($S^{\ast},C_{i},d_{i}$)
(MB-coherence($C_{i},d_{i}$)) for all nodes of the MBMP tree at levels
$i=1,\ldots,K-1$. These requirements can be considerably simplified
in two situations. According to \cite[Lemma 2]{ERC joint ksteps},
MB-ERC($S^{\ast},\hat{C},1$) implies MB-ERC($S^{\ast},C,1$) whenever
$\hat{C}\subset C\subset S^{*}$. For example, MB-ERC($S^{\ast},\emptyset,1$)
implies MB-ERC($S^{\ast},C,1$) for any $C\subset S^{*}$. More generally,
it can be shown that MB-ERC($S^{\ast},\hat{C},\hat{d}$) (MB-coherence($\hat{C},\hat{d}$))
implies MB-ERC($S^{\ast},C,d$) ((MB-coherence($C,d$))) whenever
$\hat{C}\subset C$ and $\hat{d}\leq d$. Let a node be tagged with
support $\hat{C}$, the condition $\hat{C}\subset C$ is satisfied
for any support $C$ of a descendant of such node (i.e., children,
children of children, etc.). This implies that Theorem \ref{th: MBMP recovery}
(Theorem \ref{th: MBMP recovery2}) requires MB-ERC (MB-coherence)
only at level $1$ (root node) and at nodes with a smaller number
of branches than their parents. As a concrete example, if $d_{i}=d_{1}$
for $i=1,\ldots,K-1$, Theorem \ref{th: MBMP recovery} (Theorem \ref{th: MBMP recovery2})
requires only MB-ERC($S^{*},\emptyset,d_{1}$) (MB-coherence($\emptyset,d_{1}$))
(thus requiring a similar complexity as Neuman ERC). Equivalently,
for MBMP with branch vector $\mathbf{d}=\left[d_{1},1,\ldots,1\right]$,
Theorem \ref{th: MBMP recovery} (Theorem \ref{th: MBMP recovery2})
requires MB-ERC (MB-coherence) conditions only for nodes at level
$1$ and $2$, for a total of $d_{1}+1$ conditions to be checked.
Another situation where we can simplify these conditions is in a noiseless
setup when $\operatorname{rank}\left(\mathbf{X}\right)>1$. In this
scenario, Theorem \ref{th: MBMP recovery} (Theorem \ref{th: MBMP recovery2})
requires MB-ERC (MB-coherence) only at level $i$ with $1\leq i\leq K-\operatorname{rank}\left(\mathbf{X}\right)$,
since at level $i>K-\operatorname{rank}\left(\mathbf{X}\right)$,
MBMP is guaranteed to take correct decisions thanks to the rank aware
property.

Given a matrix $\mathbf{A}$, we would like to design the number of
branches of MBMP to guarantee recovery of any $K$-sparse signals
for some targeted sparsity level $K$. An application of Theorem \ref{th: MBMP recovery2}
is to provide an upper bound on the number of branches needed by each
node of MBMP. Consider level $1$ of MBMP. By choosing $d_{1}$ as
the smallest integer such that (\ref{eq: condition d cohe}) holds
at level $1$, we guarantee that at least one node at level $2$ has
a support $C_{2}$ such that $C_{2}\subset S^{\ast}$. In general,
for each node at level $i$, we compute the refined measurement matrix
$\mathbf{\bar{A}}^{C_{i}}\triangleq\left\{ \mathbf{\bar{a}}_{g}^{C_{i}},g\notin C_{i}\right\} $,
and we select $d_{i}$ to satisfy (\ref{eq: condition d cohe}) at
level $i$. The process continues until $d_{K-1}$ is set at level
$K-1$, since nodes at level $K$ need only $d_{K}=1$ branch. Moreover,
from the discussion above, if at some node, (\ref{eq: condition d cohe})
holds with a given $d_{i}$, then, at any children of such node, the
number of branches $d_{j}$ needed to met (\ref{eq: condition d cohe})
obeys $d_{j}\leq d_{i}$. This implies that, if at some node (\ref{eq: condition d cohe})
holds with $d_{i}=1$, we can set $d_{j}=1$ branch for all children
of such node without requiring additional conditions.

\section{Numerical Results\label{Sec_Numerical}}

In this section, we present numerical results to illustrate the guarantees
obtained in Section \ref{Sec_MBMP_guaranty} and to investigate the
performance of the proposed MBMP algorithm. Although MBMP may solve
the problem (\ref{eq: Pfocus}) for any type of measurement matrix
$\mathbf{A}$, in this section we apply MBMP to perform direction-of-arrival
(DOA) estimation in a MIMO radar system where spatial compressive
sensing \cite{Rossi13} is employed. We start by introducing the MIMO
radar spatial compressive sensing setup.
\begin{figure}[ptb]
\centering{}\includegraphics[width=3.4402in,height=1.9052in]{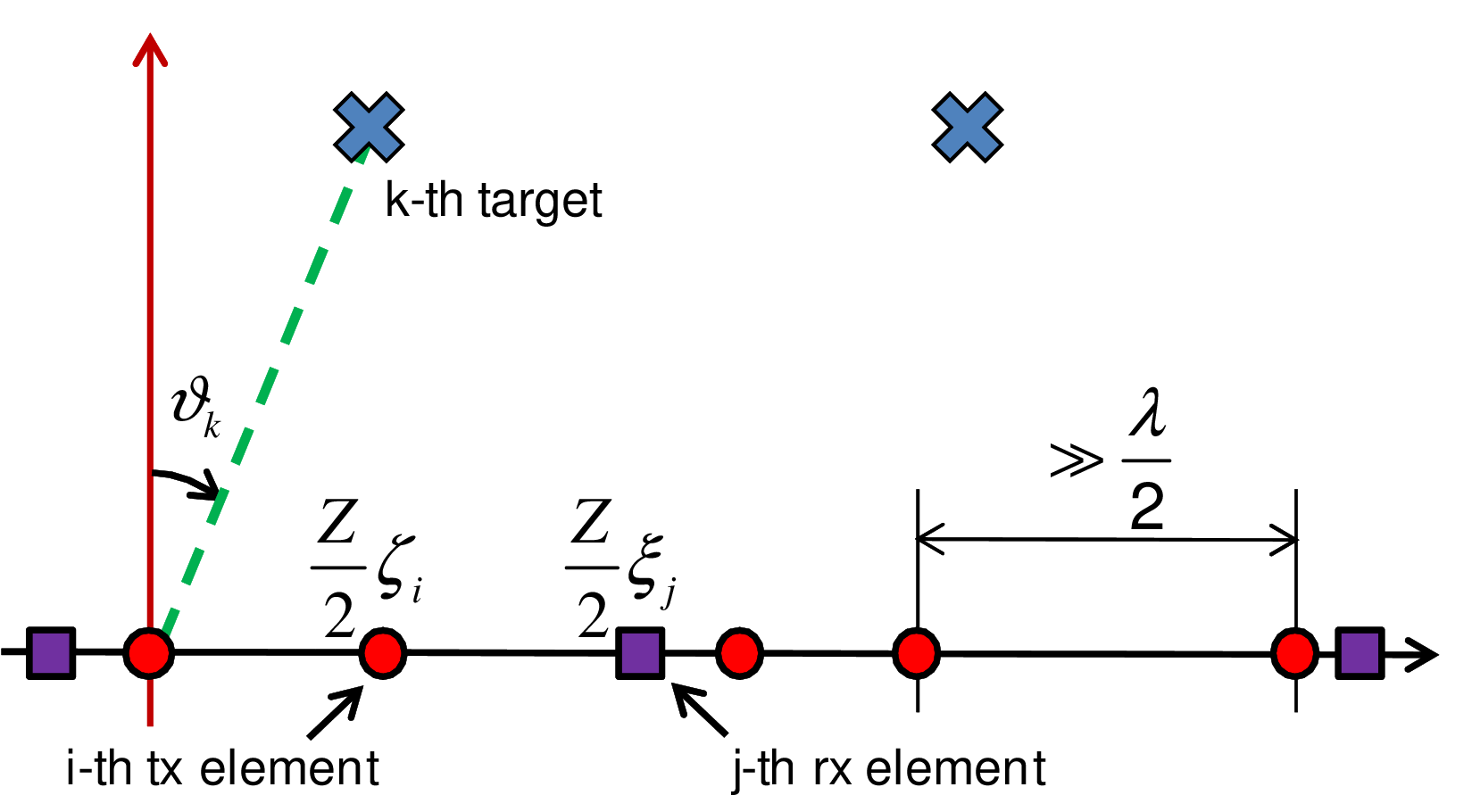}\caption{MIMO radar system model.}
\label{fig: system model}
\end{figure}

\subsection{MIMO radar setup}

We model a MIMO radar system (see Fig. \ref{fig: system model}),
where $N$ sensors collect a finite train of $l$ pulses. Each pulse
consists of $M$ orthogonal spread spectrum waveforms of length $M$
chips. Each one of the waveforms is sent by one of\ the $M$ transmitters
and returned from $K$ stationary targets. We assume that transmitters
and receivers form (possibly overlapping) linear arrays of equal aperture
$Z/2$, respectively ($Z$ is normalized in wavelength units): the
$i$-th transmitter is at position $Z\xi_{i}/2$, where $\xi_{i}\in\left[-0.5,0.5\right]$
for $i=1,\ldots,M$ on the $x$-axis; the $j$-th receiver is at position
$Z\zeta_{j}/2$, where $\zeta_{j}\in\left[-0.5,0.5\right]$ for $j=1,\ldots,N$.
The targets' positions are assumed constant over the observation interval
of $l$ pulses.

The purpose of the system is to determine the DOA angles to targets
of interest, which translate to recover the unknown signal support.
We consider targets associated with a particular range and Doppler
bin. Targets in adjacent range-Doppler bins contribute interference
to the bin of interest. The assumption of a common range bin implies
that all waveforms are received with the same time delay after transmission.
Targets are assumed in the far-field, meaning that a target's DOA
parameter $\theta\triangleq\sin\vartheta$ (where $\vartheta$ is
the DOA angle) is constant across the array. Following \cite{Rossi13},
the DOA estimation problem can be cast within a sparse localization
framework. Neglecting the discretization error, it is assumed that
the target possible locations comply with a grid of $n$ points $\phi_{1:n}$
(with $n\gg K$). By defining the $MN\times n$ matrix
\begin{equation}
\mathbf{A}\triangleq\left[\mathbf{a}\left(\phi_{1}\right),\ldots,\mathbf{a}\left(\phi_{n}\right)\right]\label{eq: MIMOradarA}
\end{equation}
where $\mathbf{a}\left(\theta\right)\triangleq\mathbf{c}\left(\theta\right)\otimes\mathbf{b}\left(\theta\right)$
with $\mathbf{b}\left(\theta\right)=\left[\exp\left(j2\pi Z\theta\zeta_{1}\right),\ldots,\exp\left(j2\pi Z\theta\zeta_{N}\right)\right]^{T}$
the receiver steering vector and $\mathbf{c}\left(\theta\right)=\left[\exp\left(j2\pi Z\theta\xi_{1}\right),\ldots,\exp\left(j2\pi Z\theta\xi_{M}\right)\right]^{T}$
the transmitter steering vector, the signal model is expressed as
(\ref{eq: Ynoisy}). In particular, the unknown matrix $\mathbf{X}\in\mathbb{C}^{n\times l}$
contains the targets locations and gains. The support of $\mathbf{X}$
corresponds to grid points with a target (see \cite{Rossi13} for
further details).

Spatial compressive sensing assumes that the elements' positions are
random variables (described by the probability density functions (pdf)
$p\left(\xi\right)$ and $p\left(\zeta\right)$). Following the setup
discussed in \cite{Rossi13}, we chose $p\left(\xi\right)$ and $p\left(\zeta\right)$
as uniform distributions, and $\phi_{1:n}$ as a uniform grid of $2/Z$-spaced
points in the range $\left[-1,1\right]$. This implies that the number
of grid points is $n=Z+1$ (columns of the measurement matrix $\mathbf{A}$).

In this section, the target gains are given by $x_{k,p}=\exp\left(-j\varphi_{k,p}\right)$,
with $\varphi_{k,p}$ drawn i.i.d., uniform over $\left[0,2\pi\right)$,
for all $k=1,...,K$ (where $K$ is the number of targets) and $p=1,\ldots,l$
(where $l$ is the number of snapshots). The noise (see (\ref{eq: Ynoisy}))
is assumed to be distributed as $\operatorname{vec}\left(\mathbf{E}\right)\sim\mathcal{CN}\left(\mathbf{0},\sigma^{2}\mathbf{I}\right)$
(where $\operatorname{vec}\left(\mathbf{\cdot}\right)$ is the vectorization
operator) and the SNR is defined as $10\log_{10}\left(\min_{k,p}\left|x_{k,p}\right|^{2}\right)-10\log_{10}\left(\sigma^{2}\right)$,
which in our setup reduces to $-10\log_{10}\left(\sigma^{2}\right)$,
since $\left|x_{k,p}\right|=1$ $\forall k,p$. From the definition
of the measurement matrix $\mathbf{A}$, its columns all have norms
equal to $\sqrt{MN}$. Throughout the numerical results, the columns
of $\mathbf{A}$ are normalized to unit norm.

\subsection{Numerical experiments}

We start by exploring the guarantee obtained in Section \ref{Sec_MBMP_guaranty},
using the MB-coherence. We investigate numerically the trade-off between
the number of measurements and number of branches $d_{1}$ at level
$1$ of MBMP (which relates to the algorithm's complexity) in order
to meet the MB-coherence($\emptyset,d_{1}$) condition (\ref{eq: condition d cohe})
at level $1$ in a noiseless setup ($\textrm{OIR}=0$), i.e.,
\begin{equation}
\max_{S,\left\vert S\right\vert =K}\left(\max_{g\in S}\left\Vert \mathbf{A}_{S}^{H}\mathbf{a}_{g}\right\Vert _{1}+d_{1}\text{\_}\max_{g\notin S}\left(\left\Vert \mathbf{A}_{S}^{H}\mathbf{a}_{g}\right\Vert _{1}\right)\right)<2,\label{eq: d cum cohe level1}
\end{equation}
where $\mathbf{\bar{A}}^{C_{1}}=\mathbf{A}$ since $C_{1}=\emptyset$.
As discussed in Section \ref{Sec_MBMP_guaranty}, condition (\ref{eq: condition d cohe})
is sufficient to guarantee the correct recovery of any $K$-sparse
signal $\mathbf{X}$ with rank $r$ using MBMP with branch vector
$\mathbf{d}\triangleq\left[d_{1},\ldots,d_{K-r},1,\ldots,1\right]$,
where $d_{i}=d_{1}$ for $i=2,\ldots,K-r$.

We generate several realizations of the MIMO radar measurement matrix
$\mathbf{A}\in\mathbb{C}^{MN\times n}$ (as defined in (\ref{eq: MIMOradarA})),
and for each realization we test whether (\ref{eq: d cum cohe level1})
holds, the probabilities of meeting the coherence condition in (\ref{eq: cohe})
and the cumulative coherence condition in (\ref{eq: cum cohe}) are
also plot as references (notice that the case $d_{1}=1$ reduces to
the Neuman ERC). Fig. \ref{fig: MBCoherenceMIMO} plots the probability
of meeting condition (\ref{eq: d cum cohe level1}) as a function
of the number of measurements $MN$ and parametrized by the number
of branches $d_{1}$. We set $K=3$, $Z=500$, and we chose $p\left(\xi\right)$
and $p\left(\zeta\right)$ as uniform distributions, and $\phi_{1:n}$
as a uniform grid of $2/Z$-spaced points in the range $\left[-1,1\right]$.
This implies that the number of grid points is $n=Z+1=501$. The main
insight of the figure is that fewer measurements are needed by the
proposed MB-coherence($\emptyset,d_{1}$) with $d_{1}>1$ compared
to previous conditions (Neuman ERC (i.e., MB-coherence($\emptyset,d_{1}$)
with $d_{1}=1$), coherence and cumulative coherence). For instance,
while Neuman ERC needs about $MN=400$ to guarantee recovery with
probability $.95$, the proposed MB-coherence($\emptyset,2$) is met
with probability $.95$ for $MN=324$ (i.e., $M=N=18$ elements).
Furthermore, as the number of branches $d_{1}$ of MBMP is increased,
fewer number of measurements is needed to guarantee recovery. For
example, only $MN=289$ measurements (i.e., $M=N=17$ elements) are
needed to guarantee MB-coherence($\emptyset,3$) with probability
$.95$, saving $6$ antenna elements with respect to the $d_{1}=1$
setup.
\begin{figure}[ptb]
\centering{}\includegraphics[width=3.4411in,height=2.6515in]{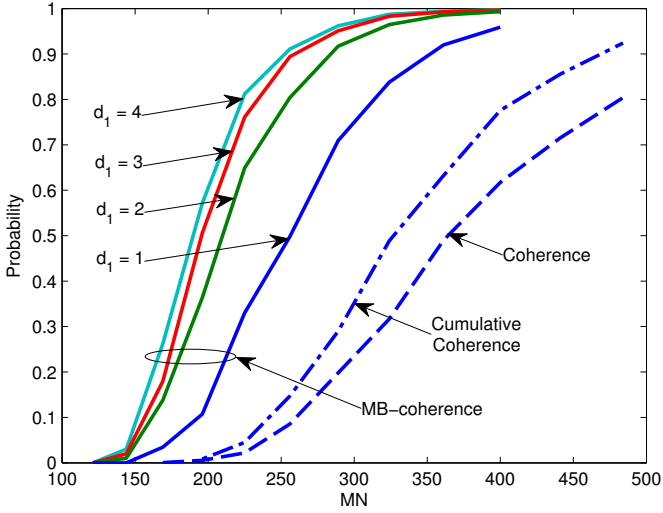}\caption{Probability of meeting condition (\ref{eq: d cum cohe level1}) as
a function of the number of measurements $MN$ for different value
of $d_{1}$. The MIMO radar measurement matrix $\mathbf{A}\in\mathbb{C}^{MN\times n}$
defined in (\ref{eq: MIMOradarA}) is employed. Signal sparsity is
$K=3$ and $n=501$.}
\label{fig: MBCoherenceMIMO}
\end{figure}

\begin{figure}[ptb]
\centering{}\includegraphics[width=3.4411in,height=2.6515in]{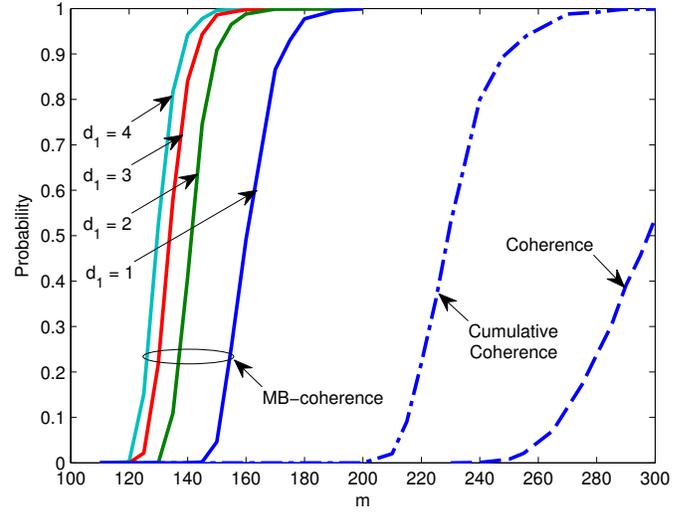}\caption{Probability of meeting condition (\ref{eq: d cum cohe level1}) as
a function of the number of measurements $m$ for different value
of $d_{1}$. The complex Gaussian measurement matrix $\mathbf{A}\in\mathbb{C}^{m\times n}$
is employed. Signal sparsity is $K=3$ and $n=501$.}
\label{fig: MBCoherenceGauss}
\end{figure}

In addition to the MIMO radar measurement matrix, we also investigate
a Gaussian measurement matrix, which has been widely studied in compressive
sensing \cite{EldarCS}. The matrix $\mathbf{A}\in\mathbb{C}^{m\times n}$
is formed by generating $mn$ i.i.d. random samples from the complex
Gaussian distribution (arranged in matrix form), and subsequently
normalizing each column of $\mathbf{A}$. In Fig. \ref{fig: MBCoherenceGauss},
we plot the probability of meeting condition (\ref{eq: d cum cohe level1})
as a function of the number of measurements $m$ for different value
of $d_{1}$. As before, we set $K=3$ and $n=501$. The advantage
of the proposed MB-coherence condition over previous conditions is
even more marked than in the MIMO radar setting, and the reduction
in the number of measurements when the number of branches $d_{1}$
of MBMP is increased can be seen from the figure. For instance, while
using Neuman ERC (i.e., MB-coherence($\emptyset,d_{1}$) with $d_{1}=1$)
we need about $m=180$ measurements to guarantee recovery with probability
$.95$, by using $d_{1}=2$ branches the same guarantee is obtained
with only $m=155$ measurements, and we can further reduce the measurements
to $m=140$, using $d_{1}=4$ branches.

The MB-coherence condition, investigated in Figures \ref{fig: MBCoherenceMIMO}
and \ref{fig: MBCoherenceGauss}, is a uniform recovery guarantee
in the sense that it guarantees recovery of any $K$-sparse signal.
Specifically, a uniform recovery guarantee certifies that, given a
fixed instantiation of the random measurement matrix $\mathbf{A}$,
all possible $K$-sparse signals are recovered correctly \cite{EldarCS}.
Uniform recovery conditions capture the worst-case behavior of a measurement
matrix $\mathbf{A}$. However, if one focuses on typical recovery,
the conditions to obtain successful (non-uniform) recovery with high
probability can be relaxed significantly, as shown in the numerical
examples below.

To investigate the typical recovery behavior of MBMP, we present numerical
results for the non-uniform recovery setting (i.e., at each realization,
the matrix $\mathbf{A}$ and the signal $\mathbf{X}$ are drawn independently
at random), and we explore the localization performance in the presence
of noise comparing MBMP with other SMV and MMV algorithms. For the
SMV setting, we implement target localization using LASSO applying
the algorithm proposed in \cite{RahutStromer}. In addition, we implement
the discrete version of beamforming (which, in the SMV setup, identifies
the support's elements as the $K$ indices $g$ that maximize $\left|\mathbf{a}_{g}^{H}\mathbf{y}\right|$),
ORMP, CoSaMP and FOCUSS \cite{focuss}. For the MMV scenario, we compare
MBMP with RA-ORMP, M-FOCUSS, and the discrete version of MUSIC (which
identifies the support's elements as the $K$ indices $g$ that maximize
$\left\Vert \mathbf{a}_{g}^{H}\mathbf{U}\right\Vert _{2}$, where
$\mathbf{U}=\operatorname{orth}\left(\mathbf{Y}\right)$ is an estimate
of the signal subspace \cite{Davies10}). As stated above, MBMP with
$\mathbf{d}=\left[1,\ldots,1\right]$ reduces to ORMP (RA-ORMP) in
the SMV (MMV) scenario.

\begin{figure}[ptb]
\centering{}\includegraphics[width=3.4411in,height=2.6878in]{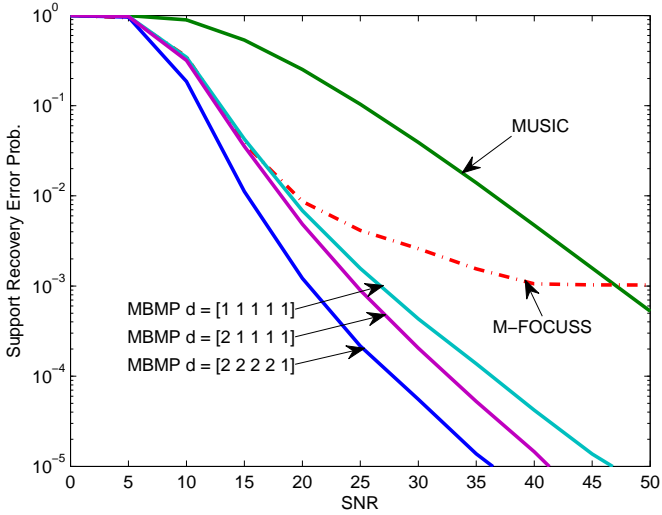}\caption{Probability of support recovery error as a function of the SNR. The
system settings are $Z=250$, $n=251$, $M=N=4$, $l=5$ and $K=5$
targets with $\left\vert x_{k,l}\right\vert =1$ for all $k$ and
$l$.}
\label{fig: M4 N4 P5 SNRvec}
\end{figure}

\begin{figure}[ptb]
\centering{}\includegraphics[width=3.4411in,height=2.7069in]{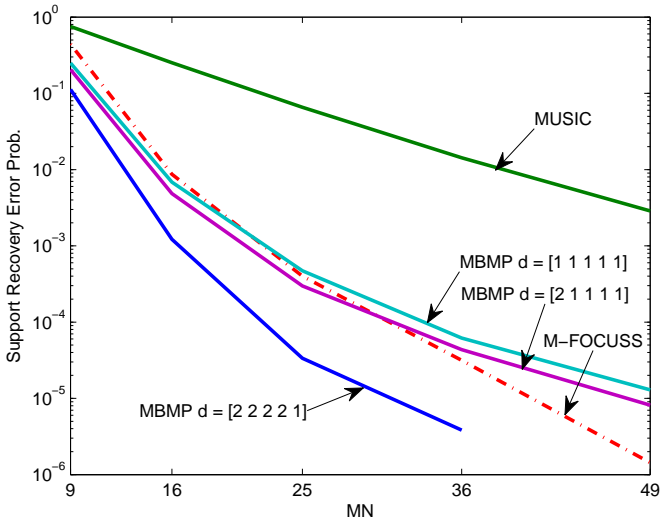}\caption{Probability of support recovery\ error as a function of the number
of rows $MN$ of $\mathbf{A}$. MMV setup ($l=5$). The system settings
are $Z=250$, $n=251$ and $K=5$ targets with $\left\vert x_{k,l}\right\vert =1$
for all $k$ and $l$. SNR is $20$ dB.}
\label{fig: P5 snr20 MNvec}
\end{figure}
We define a support recovery error event when the estimated support
does not coincide with the true one. For algorithms that return an
estimate $\mathbf{\hat{X}}$ of the sparse signal $\mathbf{X}$ (e.g.,
LASSO and M-FOCUSS), the support is identified as the $K$ largest
norm rows of the signal $\mathbf{\hat{X}}$. We further assume that
the noise variance $\sigma^{2}$ is known, since this information
is needed by LASSO and M-FOCUSS. The virtual aperture is set to $Z=250$
(thus $n=251$ grid/points), and numerical results were obtained for
$K=5$ targets.

\begin{figure}[ptb]
\centering{}\includegraphics[width=3.4411in,height=2.655in]{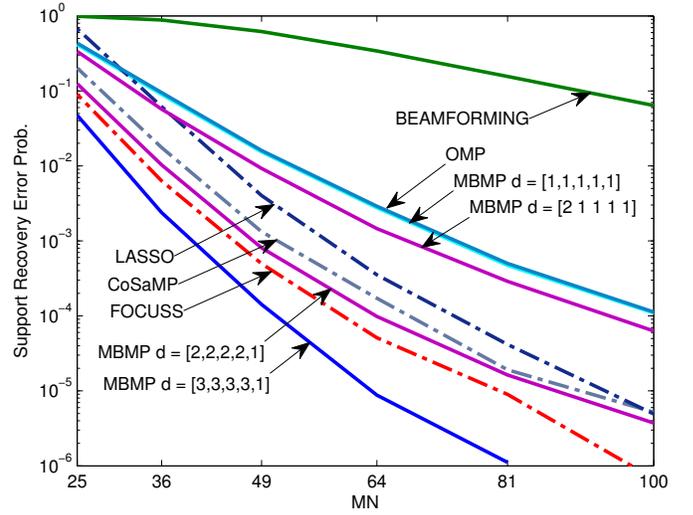}\caption{Probability of support recovery error as a function of the number
of rows $MN$ of $\mathbf{A}$. SMV setup ($l=1$). The system settings
are $Z=250$, $n=251$ and $K=5$ targets with $\left\vert x_{k}\right\vert =1$
for all $k$. SNR is $20$ dB.}
\label{fig: P1 snr20 MNvec}
\end{figure}

In Fig. \ref{fig: M4 N4 P5 SNRvec}, we address an MMV setting ($l=5$)
and we investigate the probability of support recovery error as a
function of the SNR. We set the number of antenna elements $M=N=4$.
The figure supports the theoretical findings of Section \ref{Sec_MBMP_guaranty}
that increasing the number of MBMP branches for MBMP translates into
an SNR gain. In addition, MBMP has performance superior to both M-FOCUSS
and MUSIC. The floor incurred by M-FOCUSS is due to the inability
of this method to exploit the signal subspace information (i.e., it
is not rank aware \cite{Davies10}). In addition, MBMP requires a
much smaller SNR than MUSIC: for instance, to achieve a probability
of error of $10^{-3}$, MUSIC requires SNR $=47$ dB, while MBMP with
$\mathbf{d}=\left[2,2,2,2,1\right]$ achieves the same probability
of error with just $20$ dB. This gain is ascribed to the iterative
signal support estimation performed by MBMP, which differs from the
non-iterative support estimation performed by MUSIC.

In Fig. \ref{fig: P5 snr20 MNvec}, we fix the number of snapshots
($l=5$), the SNR ($20$ dB), and we illustrate the probability of
support recovery error as a function of the number of measurements
$MN$ (number of rows of the matrix $\mathbf{A}$). We evaluate five
different element configurations: $\left(M,N\right)=\left(3,3\right)$,
$\left(4,4\right)$, $\left(5,5\right)$, $\left(6,6\right)$ and
$\left(7,7\right)$. It can be seen that, by increasing the complexity
of MBMP, the probability of error can be decreased even when we use
a limited number of antenna elements (e.g., MBMP with $\mathbf{d}=\left[2,2,2,2,1\right]$
achieves a probability of error close to $10^{-5}$ with $MN=25$).
Moreover, in all cases, MBMP performs much better than MUSIC.

In Fig. \ref{fig: P1 snr20 MNvec}, we analyze the probability of
support recovery error as a function of the number of measurements
$MN$ in an SMV setting ($l=1$). We evaluate six different configurations:
$\left(M,N\right)=\left(5,5\right)$, $\left(6,6\right)$, $\left(7,7\right)$,
$\left(8,8\right)$, $\left(9,9\right)$ and $\left(10,10\right)$,
and keep the SNR $=20$ dB. In an SMV setting, MUSIC cannot be applied
since the signal is not full-rank ($\operatorname*{rank}\left(\mathbf{X}\right)=1<K$).
In addition to MBMP and FOCUSS (the SMV version of M-FOCUSS), we performed
target DOA recovery using beamforming, LASSO and CoSaMP. From Fig.
\ref{fig: P1 snr20 MNvec} it can be seen that beamforming is not
well suited to the sparse recovery framework, incurring in a very
high probability of error as compared to sparse recovery methods.
Moreover, although in a SMV scenario the signal subspace is not available,
MBMP still provides competitive performance with respect to other
algorithm. Comparing Fig. \ref{fig: P5 snr20 MNvec} and Fig. \ref{fig: P1 snr20 MNvec},
it can be appreciated that by having multiple snapshots ($l>1$) and
using MBMP, the number of antenna elements can be dramatically reduced.
\begin{figure}[ptb]
\centering{}\includegraphics[width=3.4411in,height=2.655in]{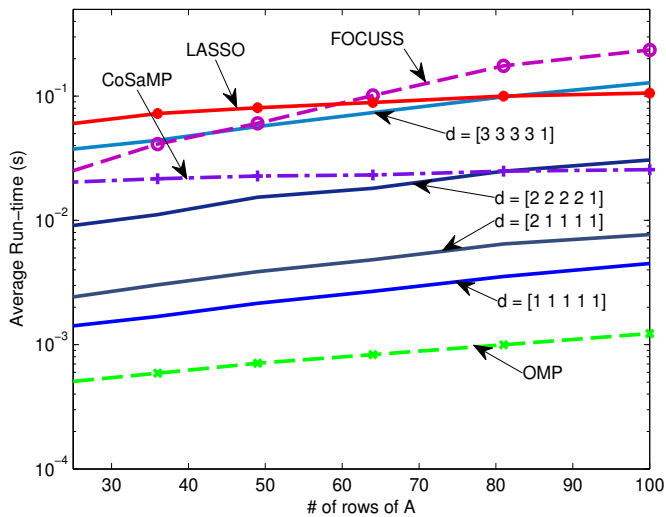}\caption{Average execution time for different CS algorithms as a function of
the number of rows of $\mathbf{A}$. SMV setup ($l=1$). The system
settings are $n=251$ and $K=5$.}
\label{fig: P1 AvgTime}
\end{figure}

Finally, we also analyze the complexity of MBMP with respect to other
CS algorithms. Fig. \ref{fig: P1 AvgTime} plots the average run-time
in seconds as a function of the number of measurements (rows of the
matrix $\mathbf{A}$) in an SMV setting ($l=1$). First, it can be
seen how, by properly setting the branch vector of MBMP, we can adjust
the MBMP complexity. Moreover, as discussed above, the figures shows
that MBMP complexity scales proportionally with the number of nodes
in the first $K$ levels of MBMP tree. In particular, the average
run-time of MBMP with $\mathbf{d}=\left[2,1,1,1,1\right]$ is slightly
less than double ($9/5$) that of MBMP with $\mathbf{d}=\left[1,1,1,1,1\right]$,
while the average run-time of MBMP with $\mathbf{d}=\left[2,2,2,2,1\right]$
is approximately $31/5=6.2$ times that of MBMP with $\mathbf{d}=\left[1,1,1,1,1\right]$.
Furthermore, although the computational complexity of MBMP is exponential
in $K$, in the scenario at hand with $K=5$, MBMP has a smaller,
or comparable, complexity to that of LASSO, FOCUSS and CoSaMP, while
providing better performance (e.g., see Fig. \ref{fig: P1 snr20 MNvec}).
We also remark that, whereas OMP complexity is smaller that ORMP (i.e.,
MBMP with $\mathbf{d}=\left[1,1,1,1,1\right]$ in an SMV setup), we
build MBMP around RA-ORMP in order to take full advantage of the rank-aware
property in a MMV setup. This is because, in radar applications, it
is common to have several snapshots and the ability to use the signal
subspace information improves performance.

\section{Conclusions\label{Sec_conc}}

We develop the MBMP algorithm for sparse recovery, and derive a sufficient
condition under which MBMP can recover any sparse signal belonging
to a given support. We then introduce the MB-coherence, and apply
it to derive a sufficient condition under which MBMP can recover \textit{any}
$K$-sparse signal. This condition enables to guarantee the success
of the proposed MBMP for dictionaries that do not satisfy previously
known conditions based on coherence or on cumulative coherence. Furthermore,
we demonstrate by numerical examples that MBMP supports trading off
measurements (e.g. antenna elements) for computational complexity.
Both theoretical guarantees and numerical results illustrate that
MBMP enables recovery with fewer measurements than other practical
algorithms.

\section{Appendix\label{Sec_appendix}}

For the sake of notation, in the Appendix we drop the superscript
$C_{i}$ from $\mathbf{\bar{A}}_{S}^{C_{i}}$ and we use $\mathbf{\bar{A}}_{S}$.

\subsection{Proof of Theorem \ref{th: MBMP recovery}}

We start by proving that, given a node at level $i$ tagged with a
correct provisional support $C_{i}\subset S^{\ast}$, if MB-ERC($S^{\ast},C_{i},d_{i}$)
in (\ref{eq: Rossi condition}) holds then at least one of the $d_{i}$
branches of the node successfully selects an index $g$ from the correct
support set $g\in S\triangleq S^{\ast}\setminus C_{i}$. We follow
similar steps as in the proof that ERC is sufficient for RA-ORMP given
in \cite{Davies10}. The only differences are: (\textit{i}) the use
the $d$\_$\max$ operator; (\textit{ii}) the use of the refined dictionary
$\mathbf{\bar{A}}^{C_{i}}\triangleq\left\{ \mathbf{\bar{a}}_{g}^{C_{i}},g\notin C_{i}\right\} $
when a provisional support $C_{i}$ is available; (\textit{iii}) the
use of the OIR to address a noisy scenario.

Similar to other MP techniques, but with the key difference of the
$d$\_$\max$ operator, in order to guarantee that at least one of
the $d_{i}$ branches of the considered node successfully selects
an atom $\mathbf{\bar{a}}_{g}^{C_{i}}$ from the remaining correct
indices $g\in S$, we require the following
\begin{equation}
\frac{d_{i}\_\max_{g\notin S}\left(\left\Vert \mathbf{U}^{H}\mathbf{\bar{a}}_{g}^{C_{i}}\right\Vert _{2}\right)}{\max_{g\in S}\left\Vert \mathbf{U}^{H}\mathbf{\bar{a}}_{g}^{C_{i}}\right\Vert _{2}}<1,\label{eq: rho}
\end{equation}
where $\mathbf{U}=\operatorname{orth}\left(\Pi_{\mathbf{A}_{C_{i}}}^{\perp}\mathbf{Y}\right)$.
Since $\mathbf{U}=\Pi_{\mathbf{A}_{S}}\mathbf{U}+\Pi_{\mathbf{A}_{S}}^{\perp}\mathbf{U}$,
by using standard norm inequalities, we can upper bound the numerator
of (\ref{eq: rho}) as 
\begin{multline}
d_{i}\_\max_{g\notin S}\left(\left\Vert \mathbf{U}^{H}\mathbf{\bar{a}}_{g}^{C_{i}}\right\Vert _{2}\right)\leq d_{i}\_\max_{g\notin S}\left(\left\Vert \mathbf{U}^{H}\Pi_{\mathbf{A}_{S}}\mathbf{\bar{a}}_{g}^{C_{i}}\right\Vert _{2}\right)\\
+\max_{g\notin S}\left\Vert \mathbf{U}^{H}\Pi_{\mathbf{A}_{S}}^{\perp}\mathbf{\bar{a}}_{g}^{C_{i}}\right\Vert _{2}.\label{eq: pippo}
\end{multline}
By using (\ref{eq: pippo}) and the definition of OIR in (\ref{eq: OIR}),
the left-hand side of (\ref{eq: rho}) can be upper bounded as
\begin{equation}
\frac{d_{i}\_\hspace{-0.03in}\max\limits _{g\notin S}\hspace{-0.03in}\left(\hspace{-0.03in}\left\Vert \mathbf{U}^{H}\mathbf{\bar{a}}_{g}^{C_{i}}\right\Vert _{2}\hspace{-0.03in}\right)}{\max\limits _{g\in S}\left\Vert \mathbf{U}^{H}\mathbf{\bar{a}}_{g}^{C_{i}}\right\Vert _{2}}\hspace{-0.03in}\leq\hspace{-0.03in}\frac{d_{i}\_\hspace{-0.03in}\max\limits _{g\notin S}\hspace{-0.03in}\left(\hspace{-0.03in}\left\Vert \mathbf{U}^{H}\Pi_{\mathbf{A}_{S}}\mathbf{\bar{a}}_{g}^{C_{i}}\right\Vert _{2}\hspace{-0.03in}\right)}{\max\limits _{g\in S}\left\Vert \mathbf{U}^{H}\mathbf{\bar{a}}_{g}^{C_{i}}\right\Vert _{2}}+\textrm{OIR}.\label{eq:pippo2}
\end{equation}
By using standard norm inequalities as in \cite{Davies10}, the first
term of the right-hand side of (\ref{eq:pippo2}) can be upper bounded
as 
\begin{equation}
\frac{d_{i}\_\max\limits _{g\notin S}\left(\left\Vert \mathbf{U}^{H}\Pi_{\mathbf{A}_{S}}\mathbf{\bar{a}}_{g}^{C_{i}}\right\Vert _{2}\right)}{\max\limits _{g\in S}\left\Vert \mathbf{U}^{H}\mathbf{\bar{a}}_{g}^{C_{i}}\right\Vert _{2}}\leq d_{i}\_\max_{g\notin S}\left(\left\Vert \mathbf{\bar{A}}_{S}^{\dag}\mathbf{\bar{a}}_{g}^{C_{i}}\right\Vert _{1}\right).\label{eq: pippo3}
\end{equation}
Using (\ref{eq: pippo3}) into inequality (\ref{eq:pippo2}), we can
conclude that, if (\ref{eq: Rossi condition}) holds, then (\ref{eq: rho})
is guaranteed to hold too. Therefore at least one of the $d_{i}$
branches of the considered node successfully selects an index $g$
from the correct support set $g\in S$.

It remains to prove that, if MB-ERC($S^{\ast},C_{i},d_{i}$) holds
for any node at level $i=1,\ldots,K-1$, then MBMP with branch vector
$\mathbf{d}=\left[d_{1},\ldots,d_{K-1},1\right]$ is guaranteed to
recover $\mathbf{X}$ from the measurements $\mathbf{Y}=\mathbf{AX+E}$.
To prove this, note that if MB-ERC($S^{\ast},C_{i},d_{i}$) holds
for any node at level $i=1,\ldots,K-1$, it follows that a chain of
correct decisions exists along the MBMP tree: MB-ERC holds for the
first node, thus at least one node at level $2$ has a correct provisional
support. Considering such node, since MB-ERC holds there, it will
select a correct index in at least one branch, and we have a node
at level $3$ with correct provisional support, and so on up to level
$K$. Finally, a node at level $K$ tagged with a correct provisional
support $C_{K}\subset S^{*}$ selects the index yielding the smallest
residual, which achieves the global optimal solution to (\ref{eq: Pfocus}),
concluding the proof.

\subsection{Proof of Theorem \ref{th: MBMP recovery2}}

We start by showing that, given a node at level $i$ tagged with a
correct provisional support $C_{i}$, the MB-coherence($C_{i},d_{i}$)
in (\ref{eq: condition d cohe}) implies MB-ERC($S^{*},C_{i},d_{i}$)
in (\ref{eq: Rossi condition}), for any support $S^{\ast}\triangleq S\cup C_{i}$
of cardinality $K$. To achieve this, we use standard arguments (e.g.,
as in \cite{Troop}) and the properties of the $d$\_$\max$ operator.
In details, by using the definition of pseudo-inverse and introducing
the $d$\_$\max$ operator, the left hand-side of (\ref{eq: Rossi condition})
can be upper bounded as
\begin{equation}
d_{i}\_\max_{g\notin S}\left(\left\Vert \mathbf{\bar{A}}_{S}^{\dag}\mathbf{\bar{a}}_{g}^{C_{i}}\right\Vert _{1}\right)\leq\frac{d_{i}\_\max\limits _{g\notin S}\left(\left\Vert \mathbf{\bar{A}}_{S}^{H}\mathbf{\bar{a}}_{g}^{C_{i}}\right\Vert _{1}\right)}{2-\max\limits _{g\in S}\left\Vert \mathbf{\bar{A}}_{S}^{H}\mathbf{\bar{a}}_{g}^{C_{i}}\right\Vert _{1}}.
\end{equation}
It follows that MB-ERC($S^{*},C_{i},d_{i}$) holds for \textit{any}
support $S^{\ast}\triangleq S\cup C_{i}$ of cardinality $K$, if
\begin{equation}
\max_{S,\left\vert S\right\vert =k}\frac{d_{i}\_\max\limits _{g\notin S}\left(\left\Vert \mathbf{\bar{A}}_{S}^{H}\mathbf{\bar{a}}_{g}^{C_{i}}\right\Vert _{1}\right)}{2-\max\limits _{g\in S}\left\Vert \mathbf{\bar{A}}_{S}^{H}\mathbf{\bar{a}}_{g}^{C_{i}}\right\Vert _{1}}<1-\textrm{OIR},
\end{equation}
where $k\triangleq K-\left\vert C_{i}\right\vert $. This can be manipulated
to obtain (\ref{eq: condition d cohe}), thus establishing that the
MB-coherence($C_{i},d_{i}$) condition (\ref{eq: condition d cohe})
implies MB-ERC($S^{\ast},C_{i},d_{i}$). The claim of the theorem
follows by invoking Theorem \ref{th: MBMP recovery}.

\subsection{Testing for MB-coherence}

We develop a practical way to find the smallest integer $d_{i}$ such
that the MB-coherence($C_{i},d_{i}$) in (\ref{eq: condition d cohe})
is met. The following proposition relates the MB-coherence condition
to an integer program, which can be solved using discrete optimization
techniques \cite{combOpt}. We denote $\mathbf{q}_{g}$ as the $g$-th
column of $\mathbf{Q}\triangleq\left\vert \left(\mathbf{\bar{A}}^{C_{i}}\right)^{H}\mathbf{\bar{A}}^{C_{i}}\right\vert $
($\left\vert \cdot\right\vert $ is the element-wise absolute value):

\begin{proposition} Let $\gamma\triangleq\frac{1}{1-\textrm{OIR}}>1$
and $k\triangleq K-\left\vert C_{i}\right\vert $. The smallest integer
$d_{i}$ such that the MB-coherence($C_{i},d_{i}$) in (\ref{eq: condition d cohe})
holds is given by the optimal objective value of
\begin{align}
 & \max_{\mathbf{s},\mathbf{y},\mathbf{z}}\text{\ \ \ \ \ }1+\sum\nolimits _{l=1}^{n}z_{l}\label{eq: MIP}\\
 & \text{s.t.\ }\left\{ \begin{array}{ll}
\left(\mathbf{q}_{j}+\gamma\mathbf{q}_{g}\right)^{T}\left(\mathbf{s+y}\right)\geq y_{j}+z_{g} & \forall g\neq j\\
\sum_{l=1}^{n}s_{l}=k-1\\
\sum_{l=1}^{n}y_{l}=1\\
y_{l}+s_{l}+z_{l}\leq1 & \forall l\\
s_{l},y_{l},z_{l}\in\left\{ 0,1\right\}  & \forall l
\end{array}\right..\nonumber 
\end{align}
 \end{proposition}

\begin{proof}Because of space limitation, we provide a sketch of
the proof. In particular, the proof follows by exploiting the one-to-one
correspondence between a set with $k$ elements out of $n$, and its
characteristic vector (i.e., a binary vector with $k$ ones and $n-k$
zeros). Let $d_{i}$ be the smallest integer such that (\ref{eq: condition d cohe})
holds. Then we have a support $S$ of cardinality $k$, an index $j\in S$,
and a set $\mathcal{G}$, such that $\left\vert \mathcal{G}\right\vert =d_{i}-1$,
$S\cap\mathcal{G}=\emptyset$, and $\left\Vert \mathbf{\bar{A}}_{S}^{H}\mathbf{\bar{a}}_{j}^{C_{i}}\right\Vert _{1}+\gamma\left\Vert \mathbf{\bar{A}}_{S}^{H}\mathbf{\bar{a}}_{g}^{C_{i}}\right\Vert _{1}\geq2$
$\forall g\in\mathcal{G}$. Given such index $j$, and the sets $S$
and $\mathcal{G}$, we can consider the associated characteristic
(binary) vectors $\mathbf{s}$, $\mathbf{y}$, and $\mathbf{z}$ (i.e.,
$y_{l}=1$ iff $l=j$; $s_{l}=1$ iff $l\in S\setminus j$; and $z_{l}=1$
iff $l\in\mathcal{G}$). Since $d_{i}=1+\left\vert \mathcal{G}\right\vert =1+\sum\nolimits _{l=1}^{n}z_{l}$,
it follows that the vectors $\mathbf{s}$, $\mathbf{y}$, and $\mathbf{z}$
maximize problem (\ref{eq: MIP}). The converse is obtained by reversing
the above argument, concluding the proof.\end{proof}

\end{document}